\newif\ifconf\conffalse

\ifconf
\documentclass[twoside,leqno,twocolumn]{article}
\usepackage{ltexpprt}

\else
\documentclass[11pt,letterpaper]{article}
\usepackage{amsthm}
\usepackage[margin=1in,dvips]{geometry}
\fi

\usepackage{graphicx}
\usepackage{amssymb}
\usepackage{amsmath}
\usepackage{float}
\usepackage{enumerate}

\iffalse
\newtheorem{definition}{Definition}[section]
\newtheorem{remark}{Remark}[section]
\else
\newtheorem{theorem}{Theorem}[section]
\newtheorem{lemma}[theorem]{Lemma}
\newtheorem{corollary}[theorem]{Corollary}
\newtheorem{definition}[theorem]{Definition}

\newtheorem{remark}[theorem]{Remark}
\newtheorem{fact}[theorem]{Fact}
\fi

\newcommand{\abs}[1]{\left|#1\right|}

\newcommand{\norm}[2]{\left \lVert#2\right \rVert_{#1}}

\newcommand{\gaplinf}{\mathsf{Gap}\ell_{\infty}}

\newcommand{\multiLinf}{\mathsf{Multi}\ell_{\infty}}

{\makeatletter
 \gdef\xxxmark{%
   \expandafter\ifx\csname @mpargs\endcsname\relax 
     \expandafter\ifx\csname @captype\endcsname\relax 
       \marginpar{xxx}
     \else
       xxx 
     \fi
   \else
     xxx 
   \fi}
 \gdef\xxx{\@ifnextchar[\xxx@lab\xxx@nolab}
 \long\gdef\xxx@lab[#1]#2{{\bf [\xxxmark #2 ---{\sc #1}]}}
 \long\gdef\xxx@nolab#1{{\bf [\xxxmark #1]}}
}

\DeclareMathOperator*{\E}{\mathbb{E}}
\def\R{\mathbb{R}}

\def\eps{\epsilon}

\begin{document}

\title{Lower Bounds for Adaptive Sparse Recovery}

\author{Eric Price\\MIT\\ecprice@mit.edu \and David P. Woodruff\\IBM Almaden\\dpwoodru@us.ibm.com}
\date{}
\maketitle
\thispagestyle{empty}
\setcounter{page}{0}

\begin{abstract}
  We give lower bounds for the problem of stable sparse recovery from
  \emph{adaptive} linear measurements.  In this problem, one would
  like to estimate a vector $x \in \R^n$ from $m$ linear measurements
  $A_1x, \dotsc, A_mx$.  One may choose each vector $A_i$ based on
  $A_1x, \dotsc, A_{i-1}x$, and must output $\hat{x}$ satisfying
  \[
  \norm{p}{\hat{x} - x} \leq (1 + \eps) \min_{k\text{-sparse } x'} \norm{p}{x - x'}
  \]
  with probability at least $1-\delta>2/3$, for some $p \in \{1,2\}$.  For
  $p=2$, it was recently shown that this is possible with $m =
  O(\frac{1}{\eps}k \log \log (n/k))$, while nonadaptively it requires
  $\Theta(\frac{1}{\eps}k \log (n/k))$.  It is also known that even
  adaptively, it takes $m = \Omega(k/\eps)$ for $p = 2$. For $p = 1$,
  there is a non-adaptive upper bound of
  $\widetilde{O}(\frac{1}{\sqrt{\eps}} k\log n)$.  We show:
  \begin{itemize}
  \item For $p=2$, $m = \Omega(\log \log n)$. This is tight for $k =
    O(1)$ and constant $\eps$, and shows that the $\log \log n$
    dependence is correct.
  \item If the measurement vectors are chosen in $R$ ``rounds'', then
    $m = \Omega(R \log^{1/R} n)$.  For constant $\eps$, this matches
    the previously known upper bound up to an $O(1)$ factor in $R$.
  \item For $p=1$, $m = \Omega(k/(\sqrt{\eps} \cdot \log k/\eps))$. This shows that adaptivity
    cannot improve more than logarithmic factors, providing the analogue of the 
    $m = \Omega(k/\eps)$ bound for $p = 2$. 
  \end{itemize}
\end{abstract} 
\newpage

\section{Introduction}

\emph{Compressed sensing} or \emph{sparse recovery} studies the
problem of solving underdetermined linear systems subject to a
sparsity constraint.  It has applications to a wide variety of fields,
including data stream algorithms~\cite{M05}, medical or geological
imaging~\cite{CRT06,D06}, and genetics testing~\cite{SAZ}.  The
approach uses the power of a \emph{sparsity} constraint: a vector $x'$
is \emph{$k$-sparse} if at most $k$ coefficients are non-zero.  A
standard formulation for the problem is that of \emph{stable sparse
  recovery}: we want a distribution $\mathcal{A}$ of matrices $A \in
\R^{m \times n}$ such that, for any $x \in \R^n$ and with probability
$1-\delta>2/3$ over $A \in \mathcal{A}$, there is an algorithm to recover
$\hat{x}$ from $Ax$ with
\begin{align}\label{eq:lplp}
  \norm{p}{\hat{x} - x} \leq (1 + \eps) \min_{k\text{-sparse } x'} \norm{p}{x - x'}
\end{align}
for some parameter $\eps > 0$ and norm $p$.  We refer to the elements
of $Ax$ as \emph{measurements}.  We say Equation~\eqref{eq:lplp}
denotes \emph{$\ell_p/\ell_p$ recovery}.

The goal is to minimize the number of measurements while still
allowing efficient recovery of $x$.  This problem has recently been
largely closed: for $p=2$, it is known that $m = \Theta(\frac{1}{\eps}
k \log (n/k))$ is tight~(upper bounds in \cite{CRT06,GLPS}, lower
bounds in \cite{PW11,CD11}), and for $p=1$ it is known that $m =
\widetilde{O}(\frac{1}{\sqrt{\eps}}k\log n)$ and $m =
\widetilde{\Omega}(\frac{k}{\sqrt{\eps}})$~\cite{PW11} (recall that
$\widetilde{O}(f)$ means $O(f \log^c f)$ for some constant c, and
similarly $\widetilde{\Omega}(f)$ means $\Omega(f / \log^c f)$).

In order to further reduce the number of measurements, a number of recent works
have considered making the measurements
\emph{adaptive}~\cite{JXC,CHNR,HCN,HBCN,MSW,AWZ,IPW11}.
In this setting, one may choose each row of the matrix after seeing
the results of previous measurements.  More generally, one may split
the adaptivity into $R$ ``rounds'', where in each round $r$ one
chooses $A^r \in \R^{m_r \times n}$ based on $A^1x, \dotsc, A^{r-1}x$.
At the end, one must use $A^1x, \dotsc, A^Rx$ to output $\hat{x}$
satisfying Equation~\eqref{eq:lplp}.  We would still like to minimize
the total number of measurements $m = \sum m_i$.  In the $p=2$
setting, it is known that for arbitrarily many rounds
$O(\frac{1}{\eps}k \log \log (n/k))$ measurements suffice, and for
$O(r \log^*k)$ rounds $O(\frac{1}{\eps}k r\log^{1/r} (n/k))$
measurements suffice~\cite{IPW11}.  

Given these upper bounds, two natural questions arise: first, is the
improvement in the dependence on $n$ from $\log (n/k)$ to $\log \log
(n/k)$ tight, or can the improvement be strengthened?  Second, can
adaptivity help by more than a logarithmic factor, by improving the
dependence on $k$ or $\eps$?

A recent lower bound showed that $\Omega(k/\eps)$ measurements are
necessary in a setting essentially equivalent to the $p = 2$
case~\cite{ACD11}\footnote{Both our result and their result apply in
  both settings.  See Appendix~\ref{app:noisekind} for a more detailed
  discussion of the relationship between the two settings.}.  Thus, they
answer the second question in the negative for $p=2$.
Their techniques rely on special properties of the
$2$-norm; namely, that it is a rotationally invariant inner product
space and that the Gaussian is both $2$-stable and a maximum entropy
distribution.  Such techniques do not seem useful for proving lower
bounds for $p = 1$.

\paragraph{Our results.} For $p=2$, we show that any adaptive sparse
recovery scheme requires $\Omega(\log \log n)$ measurements, or
$\Omega(R \log^{1/R} n)$ measurements given only $R$ rounds.  For $k =
O(1)$, this matches the upper bound of~\cite{IPW11} up to an
$O(1)$ factor in $R$.  It thus shows that the $\log \log n$
term in the adaptive bound is necessary.

For $p=1$, we show that any adaptive sparse recovery scheme requires
$\widetilde{\Omega}(k/\sqrt{\eps})$ measurements. This shows that
adaptivity can only give $\text{polylog}(n)$ improvements, even for $p
= 1$.  Additionally, our bound of $\Omega(k / (\sqrt{\eps} \cdot \log
(k / \sqrt{\eps})))$ improves the previous \emph{non-adaptive} lower
bound for $p=1$ and small $\eps$, which lost an additional $\log k$
factor~\cite{PW11}.

\paragraph{Related work.} Our work draws on the lower bounds for
non-adaptive sparse recovery, most directly~\cite{PW11}.

The main previous lower bound for adaptive sparse recovery gets $m =
\Omega(k/\eps)$ for $p=2$~\cite{ACD11}.  They consider going down a
similar path to our $\Omega(\log \log n)$ lower bound, but ultimately
reject it as difficult to bound in the adaptive setting.  Combining
their result with ours gives a $\Omega(\frac{1}{\eps}k + \log \log n)$
lower bound, compared with the $O(\frac{1}{\eps}k \cdot \log \log n)$
upper bound. The techniques in their paper do not imply any bounds for
the $p = 1$ setting.

For $p=2$ in the special case of adaptive Fourier measurements (where
measurement vectors are adaptively chosen from among $n$ rows of the
Fourier matrix),~\cite{HIKP12} shows $\Omega(k \log (n/k) / \log \log
n)$ measurements are necessary.  In this case the main difficulty with
lower bounding adaptivity is avoided, because all measurement rows are
chosen from a small set of vectors with bounded $\ell_\infty$ norm;
however, some of the minor issues in using~\cite{PW11} for an adaptive
bound were dealt with there.

\paragraph{Our techniques.}  We use very different techniques for our
two bounds.

To show $\Omega(\log \log n)$ for $p=2$, we reduce to the information
capacity of a Gaussian channel.  We consider recovery of the vector $x
= e_{i^*} + w$, for $i^* \in [n]$ uniformly and $w \sim N(0,
I_n/\Theta(n))$.  Correct recovery must find $i^*$, so the mutual
information $I(i^*; Ax)$ is $\Omega(\log n)$.  On the other hand, in
the nonadaptive case~\cite{PW11} showed that each measurement $A_jx$
is a power-limited Gaussian channel with constant signal-to-noise
ratio, and therefore has $I(i^*; A_jx) = O(1)$.  Linearity gives that
$I(i^*; Ax) = O(m)$, so $m = \Omega(\log n)$ in the nonadaptive case.
In the adaptive case, later measurements may ``align'' the row $A_j$
with $i^*$, to increase the signal-to-noise ratio and extract more
information---this is exactly how the upper bounds work.  To deal with
this, we bound how much information we can extract as a function of
how much we know about $i^*$.  In particular, we show that given a
small number $b$ bits of information about $i^*$, the posterior
distribution of $i^*$ remains fairly well ``spread out''.  We then
show that any measurement row $A_j$ can only extract $O(b+1)$ bits
from such a spread out distribution on $i^*$.  This shows that the
information about $i^*$ increases at most exponentially, so
$\Omega(\log \log n)$ measurements are necessary.

To show an $\widetilde{\Omega}(k/\sqrt{\eps})$ bound for $p=1$, 
we first establish a lower bound on the multiround distributional 
communication complexity
of a two-party communication problem that we call $\multiLinf$, for a 
distribution tailored to our application. 
We then show how to use an adaptive $(1+\eps)$-approximate 
$\ell_1/\ell_1$ sparse
recovery scheme $\mathcal{A}$ 
to solve the communication problem $\multiLinf$, modifying the 
general framework
of \cite{PW11} for connecting non-adaptive schemes to communication 
complexity in
order to now support adaptive schemes. 
By the communication lower bound for $\multiLinf$, 
we obtain a lower bound on the number of measurements required of $\mathcal{A}$. 

In the $\gaplinf$ problem, the two players are given $x$ and $y$ respectively,
and they want to approximate $\|x-y\|_{\infty}$ given the promise that all entries
of $x-y$ are small in magnitude or there is a single large entry. The $\multiLinf$
problem consists of solving multiple independent instances of $\gaplinf$ in parallel. Intuitively,
the sparse recovery algorithm needs to determine if there are entries of $x-y$ that are large, which
corresponds to solving multiple instances of $\gaplinf$. 
We prove a multiround direct sum theorem for a distributional version of $\gaplinf$, 
thereby giving a distributional lower bound for $\multiLinf$. A direct sum theorem for $\gaplinf$
has been used before for proving lower bounds
for non-adaptive schemes \cite{PW11}, but was limited to a bounded
number of rounds due to the use of a bounded round theorem in
communication complexity \cite{br11}.
We instead use the information complexity framework \cite{BJKS04} to 
lower bound the conditional mutual information between the inputs to $\gaplinf$
and the transcript of any correct protocol for $\gaplinf$ under a certain input distribution, and prove a 
direct sum theorem for solving $k$ instances of this problem. 
We need to condition on ``help variables'' in the mutual information which enable the players
to embed single instances of $\gaplinf$ into $\multiLinf$ in a way in which the players can use
a correct protocol on our input distribution for $\multiLinf$ 
as a correct protocol on
our input distribution for $\gaplinf$; these help variables are 
in addition to help variables used for proving
lower bounds for $\gaplinf$, which
is itself proved using information complexity. We also look 
at the conditional mutual information with respect to an input distribution 
which doesn't immediately fit into the information complexity framework. 
We relate the conditional information of the transcript with respect to this
distribution to that with respect to a more standard distribution. 
%
%
\section{Notation}

We use lower-case letters for fixed values and upper-case letters for
random variables.  We use $\log x$ to denote $\log_2 x$, and $\ln x$
to denote $\log_e x$.  For a discrete random variable $X$ with
probability $p$, we use $H(X)$ or $H(p)$ to denote its entropy
\[
H(X) = H(p) = \sum - p(x) \log p(x).
\]
For a continuous random variable $X$ with pdf $p$, we use $h(X)$ to
denote its differential entropy
\[
h(X) = \int_{x \in X} -p(x) \log p(x) dx.
\]
Let $y$ be drawn from a random variable $Y$.  Then $(X \mid y) = (X
\mid Y = y)$ denotes the random variable $X$ conditioned on $Y = y$.
We define $h(X \mid Y) = \E_{y\sim Y} h(X \mid y)$.  The mutual
information between $X$ and $Y$ is denoted $I(X; Y) = h(X) - h(X \mid
Y)$.

For $p \in \R^n$ and $S \subseteq [n]$, we define $p_S \in \R^n$ to
equal $p$ over indices in $S$ and zero elsewhere.

We use $f \lesssim g$ to denote $f = O(g)$.

\section{Tight lower bound for $p=2, k=1$ }

We may assume that the measurements are orthonormal, since this can be
performed in post-processing of the output, by multiplying $Ax$ on the
left to orthogonalize $A$.  We will give a lower bound for the
following instance:

Alice chooses random $i^* \in [n]$ and i.i.d.\ Gaussian noise $w \in \R^n$
with $\E[\norm{2}{w}^2] = \sigma^2 = \Theta(1)$, then sets $x =
e_{i^*} + w$.  Bob performs $R$ rounds of adaptive measurements on
$x$, getting $y^r = A^rx = (y^r_1, \dotsc, y^r_{m_r})$ in each round
$r$.  Let $I^*$ and $Y^r$ denote the random variables from which $i^*$
and $y^r$ are drawn, respectively.  We will bound $I(I^*; Y^1, Y^2,
\dotsc, Y^r)$.

We may assume Bob is deterministic, since we are giving a lower bound
for a distribution over inputs -- for any randomized Bob that succeeds
with probability $1-\delta$, there exists a choice of random seed such
that the corresponding deterministic Bob also succeeds with
probability $1-\delta$.

First, we give a bound on the information received from any single
measurement, depending on Bob's posterior distribution on $I^*$ at
that point:

\begin{lemma}\label{lemma:oneround}
  Let $I^*$ be a random variable over $[n]$ with probability
  distribution $p_i = \Pr[I^*=i]$, and define
  \[
  b = \sum_{i=1}^n p_i \log (n p_i).
  \]
  Define $X = e_{I^*} + N(0, I_n\sigma^2/n)$.  Consider any fixed
  vector $v \in \R^n$ independent of $X$ with $\norm{2}{v} = 1$, and
  define $Y = v \cdot X$.  Then
  \[
  I(v_{I^*}; Y) \leq C(b + 1)
  \]
  for some constant $C$.
\end{lemma}
\begin{proof}
  Let $S_i = \{j \mid 2^{i} \leq np_j < 2^{i+1}\}$ for $i > 0$ and
  $S_{0} = \{i \mid np_i < 2\}$.  Define $t_i = \sum_{j \in S_i} p_j =
  \Pr[I^* \in S_i]$.  Then
  \begin{align*}
    \sum_{i=0}^{\infty} i t_i &= \sum_{i > 0} \sum_{j \in S_i} p_j \cdot i\\
    &\leq \sum_{i > 0} \sum_{j \in S_i} p_j\log (np_j)\\
    &= b - \sum_{j \in S_{0}} p_j \log (np_j)\\
    &\leq b - t_{0} \log (n t_{0} / \abs{S_{0}})\\
    &\leq b + \abs{S_{0}} / (ne)
  \end{align*}
  using convexity and minimizing $x \log ax$ at $x = 1/(ae)$.  Hence
  \begin{align}\label{eq:it_i}
    \sum_{i=0}^{\infty} i t_i < b + 1
  \end{align}

  Let $W = N(0, \sigma^2/n)$.  For any measurement vector $v$, let $Y
  = v \cdot X \sim v_{I^*} + W$.  Let $Y_i = (Y \mid I^* \in S_i)$.
  Because $\sum v_j^2 = 1$,
  \begin{align}\label{eq:Yibound}
    \E[Y_i^2] &= \sigma^2/n + \sum_{j \in S_i} v_j^2 p_j/t_i 
\leq
    \sigma^2/n + \norm{\infty}{p_{S_i}}/t_i \leq \sigma^2/n +
    2^{i+1}/(nt_i).
  \end{align}
  Let $T$ be the (discrete) random variable denoting the $i$ such that
  $I^* \in S_i$.  Then $Y$ is drawn from $Y_T$, and $T$ has
  probability distribution $t$.  Hence
  \begin{align*}
    h(Y) &\leq h((Y, T))\\
    &= H(T) + h(Y_T \mid T)\\
    &= H(t) + \sum_{i \geq 0} t_i h(Y_i)\\
    &\leq H(t) + \sum_{i \geq 0} t_i h(N(0, \E[Y_i^2]))
  \end{align*}
  because the Gaussian distribution maximizes entropy subject to a
  power constraint.  Using the same technique as the Shannon-Hartley
  theorem,
  \begin{align*}
    I(v_{I^*}, Y)
= I(v_{I^*}; v_{I^*} + W) 
    &= h(v_{I^*} + W) - h(v_{I^*}+W \mid v_{I^*})\\
    &= h(Y) - h(W)\\
    &\leq H(t) + \sum_{i \geq 0} t_i (h(N(0, \E[Y_i^2])) - h(W))\\
    &= H(t) + \frac{1}{2}\sum_{i \geq 0} t_i \ln(\frac{\E[Y_i^2]}{\E[W^2]})
  \end{align*}
  and hence by Equation~\eqref{eq:Yibound},
  \begin{align}\label{eq:Iint_i}
    I(v_{I^*}; Y) \leq H(t) + \frac{\ln 2}{2}\sum_{i \geq 0} t_i \log(1 + \frac{2^{i+1}}{t_i\sigma^2}).
  \end{align}

  All that requires is to show that this is $O(1 + b)$.  Since $\sigma
  = \Theta(1)$, we have
  \begin{align}
    \sum_i t_i \log(1 + \frac{2^i}{\sigma^2t_i})
    &\leq \log(1 + 1/\sigma^2) + \sum_i t_i \log(1 + \frac{2^i}{t_i})\notag\\
    &\leq O(1) + \sum_{i} t_i \log(1 + 2^i) +\phantom{}
\sum_{i} t_i \log(1 + 1/t_i).\label{eq:splitting}
  \end{align}
  Now, $\log(1 + 2^i) \lesssim i$ for $i > 0$ and is $O(1)$ for $i =
  0$, so by Equation~\eqref{eq:it_i},
  \[
  \sum_{i} t_i \log(1 + 2^i) \lesssim 1 + \sum_{i > 0} i t_i < 2 + b.
  \]
  Next, $\log(1 + 1/t_i) \lesssim \log(1/t_i)$ for $t_i \leq 1/2$, so
  \begin{align*}
    \sum_{i} t_i \log(1 + 1/t_i) &\lesssim \sum_{i \mid t_i \leq 1/2} t_i
    \log (1/t_i) + \sum_{i \mid t_i > 1/2} 1
    \leq H(t) + 1.
  \end{align*}
  Plugging into Equations~\eqref{eq:splitting} and~\eqref{eq:Iint_i},
  \begin{align}\label{eq:nasty_intermediate}
    I(v_{I^*}, Y) \lesssim 1 + b +
    H(t).
  \end{align}
  To bound $H(t)$, we consider the partition $T_+ = \{i \mid t_i >
  1/2^i\}$ and $T_- = \{i \mid t_i \leq 1/2^i\}$.  Then
  \begin{align*}
    H(t) &= \sum_i t_i \log (1/t_i)\\
    &\leq \sum_{i \in T_+} i t_i + \sum_{t \in T_-} t_i \log (1/t_i)\\
    &\leq 1 + b + \sum_{t \in T_-} t_i \log (1/t_i)
  \end{align*}
  But $x \log (1/x)$ is increasing on $[0, 1/e]$, so
  \begin{align*}
    \sum_{t \in T_-} t_i \log (1/t_i) &\leq t_0 \log (1/t_0) + t_1 \log
    (1/t_1) +
\sum_{i \geq 2} \frac{1}{2^i} \log (1/2^i) 
\leq 2/e + 3/2 = O(1)
  \end{align*}
  and hence $H(t) \leq b + O(1)$.  Combining with
  Equation~\eqref{eq:nasty_intermediate} gives that
  \[
  I(v_{I^*}; Y) \lesssim b + 1
  \]
  as desired.
\end{proof}

\begin{theorem}
  Any scheme using $R$ rounds with number of measurements $m_1, m_2,
  \dotsc, m_R > 0$ in each round has
  \[
  I(I^*; Y^1, \dotsc, Y^R) \leq C^R\prod_i m_i
  \]
  for some constant $C > 1$.
\end{theorem}
\begin{proof}
  Let the signal in the absence of noise be $Z^r = A^re_{I^*} \in
  \R^{m_r}$, and the signal in the presence of noise be $Y^r =
  A^r(e_{I^*} + N(0, \sigma^2 I_n / n)) = Z^r + W^r$ where $W^r = N(0,
  \sigma^2I_{m_r}/n)$ independently.  In round $r$, after observations
  $y^1, \dotsc, y^{r-1}$ of $Y^1, \dotsc, Y^{r-1}$, let $p^r$ be the
  distribution on $(I^* \mid y^1, \dotsc, y^{r-1})$.  That is, $p^r$
  is Bob's posterior distribution on $I^*$ at the beginning of round
  $r$.

  We define
  \begin{align*}
  b_r &= H(I^*) - H(I^* \mid y^1, \dotsc, y^{r-1})\\
  &= \log n - H(p^r)\\
  &= \sum p^r_i \log (n p^r_i).
  \end{align*}
  Because the rows of $A^r$ are deterministic given $y^1, \dotsc,
  y^{r-1}$, Lemma~\ref{lemma:oneround} shows that any single
  measurement $j \in [m_r]$ satisfies
  \[
  I(Z^r_j; Y^r_j  \mid y^1, \dotsc, y^{r-1}) \leq C(b_r+1).
  \]
  for some constant $C$.  Thus by Lemma~\ref{lemma:splitentropy}
  \[
  I(Z^r; Y^r  \mid y^1, \dotsc, y^{r-1}) \leq Cm_r(b_r+1).
  \]
  There is a Markov chain $(I^* \mid y^1, \dotsc, y^{r-1}) \to (Z^r
  \mid y^1, \dotsc, y^{r-1}) \to (Y^r \mid y^1, \dotsc, y^{r-1})$, so
  \begin{align*}
    I(I^*; Y^r \mid y^1, \dotsc, y^{r-1})
\leq I(Z^r; Y^r \mid y^1, \dotsc, y^{r-1}) 
\leq Cm_r(b_r + 1).
  \end{align*}
  We define $B_r = I(I^*; Y^1, \dotsc, Y^{r-1}) = \E_{y} b_r$.  Therefore
  \begin{align*}
    B_{r+1}
    &= I(I^*; Y^1, \dotsc, Y^r) \\
    &= I(I^*; Y^1, \dotsc, Y^{r-1}) + I(I^*; Y^r \mid Y^1, \dotsc, Y^{r-1})\\
    &= B_r + \E_{y^1, \dotsc, y^{r-1}} I(I^*; Y^r \mid y^1, \dotsc, y^{r-1})\\
    &\leq B_r + Cm_r \E_{y^1, \dotsc, y^{r-1}} (b_r+1) \\
    &= (B_r+1)(Cm_r + 1) - 1\\
    &\leq C'm_r (B_r + 1)
  \end{align*}
  for some constant $C'$.  Then for some constant $D \geq C'$,
  \[
  I(I^*; Y^1, \dotsc, Y^R) = B_{R+1} \leq D^R\prod_i m_i
  \]
  as desired.
\end{proof}

\begin{corollary}
  Any scheme using $R$ rounds with $m$ measurements has
  \[
  I(I^*; Y^1, \dotsc, Y^R) \leq (Cm/R)^R
  \]
  for some constant $C$.  Thus for sparse recovery, $m =
  \Omega(R\log^{1/R} n)$.  Minimizing over $R$, we find that $m =
  \Omega(\log \log n)$ independent of $R$.
\end{corollary}
\begin{proof}
  The equation follows from the AM-GM inequality.  Furthermore, our
  setup is such that Bob can recover $I^*$ from $Y$ with large
  probability, so $I(I^*; Y) = \Omega(\log n)$; this was formally
  shown in Lemma 6.3 of~\cite{HIKP12} (modifying Lemma 4.3
  of~\cite{PW11} to adaptive measurements and $\eps = \Theta(1)$).
  The result follows.
\end{proof}

\section{Lower bound for dependence on $k$ and $\eps$ for $\ell_1/\ell_1$}
\label{sec:eps}
%
In Section
\ref{sec:cc1} we establish a new lower bound on the communication complexity
of a two-party communication problem that we call $\multiLinf$. 
In Section \ref{sec:cc2} we
then show how to use an adaptive $(1+\eps)$-approximate $\ell_1/\ell_1$ sparse
recovery scheme $\mathcal{A}$ 
to solve the communication problem $\multiLinf$. 
By the communication lower bound in Section \ref{sec:cc1}, 
we obtain a lower bound on the number of measurements required of $\mathcal{A}$. 

\subsection{Direct sum for distributional $\ell_{\infty}$}\label{sec:cc1}
We assume basic familiarity with communication complexity; see the
textbook of Kushilevitz and Nisan \cite{kn97} for further background.
Our reason for using communication complexity is to prove lower
bounds, and we will do so by using information-theoretic arguments. We
refer the reader to the thesis of Bar-Yossef \cite{BarYossefThesis}
for a comprehensive introduction to information-theoretic arguments
used in communication complexity.

We consider two-party randomized communication complexity. There are two parties,
Alice and Bob, with input vectors $x$ and $y$ respectively, and their goal is
to solve a promise problem $f(x,y)$. The parties have private randomness. 
The communication cost of a protocol is
its maximum transcript length, over all possible inputs and random coin tosses.
The randomized communication complexity $R_{\delta}(f)$ is 
the minimum communication cost of a randomized protocol $\Pi$ which for every input
$(x,y)$ outputs $f(x,y)$ with probability at least $1-\delta$ (over the random
coin tosses of the parties). We also study the distributional
complexity of $f$, in which the parties are deterministic and
the inputs $(x,y)$ are drawn from distribution $\mu$, and a protocol is
correct if it succeeds with probability at least $1-\delta$ in outputting $f(x,y)$,
where the probability is now taken over $(x,y) \sim \mu$. We define $D_{\mu, \delta}(f)$
to be the minimum communication cost of a correct protocol $\Pi$. 

We consider the following promise problem 
$\gaplinf^B$, where $B$ is a parameter, 
which was studied in \cite{ss02,BJKS04}. The
inputs are pairs $(x,y)$ of $m$-dimensional vectors, with
$x_i, y_i \in \{0, 1, 2, \ldots, B\}$ for all $i \in [m]$, with the promise that 
$(x,y)$ is one of the following types of instance:
\begin{itemize}
\item NO instance: for all $i$, $|x_i-y_i| \in \{0,1\}$, or
\item YES instance: there is a unique $i$ for which
  $|x_i-y_i| = B$, and for all $j \neq i$, $|x_j-y_j| \in \{0,1\}$.
\end{itemize}
The goal of a protocol is to decide which of the two cases (NO or YES) the input
is in. 

Consider the distribution $\sigma$: for each $j \in
[m]$, choose a random pair $(Z_j, P_j) \in \{0, 1, 2, \ldots, B\} \times \{0,1\}
\setminus \{(0,1), (B, 0)\}$. If $(Z_j, P_j) = (z, 0)$, then $X_j = z$ and 
$Y_j$ is uniformly distributed in $\{z, z+1\}$; if $(Z_j, P_j) = (z, 1)$, then $Y_j = z$
and $X_j$ is uniformly distributed on $\{z-1, z\}$. 
Let $Z = (Z_1, \ldots, Z_m)$ and $P = (P_1, \ldots, P_m)$. 
Next choose a random coordinate $S \in [m]$. 
For coordinate $S$, replace $(X_{S}, Y_{S})$ with a uniform element of 
$\{(0,0), (0,B)\}$.  Let $X = (X_1, \ldots, X_m)$ and $Y =
(Y_1, \ldots, Y_m)$.

Using similar arguments to those in \cite{BJKS04}, 
we can show that there are positive, sufficiently small 
constants $\delta_0$ and $C$ so that
for any randomized protocol $\Pi$ which succeeds with probability at
least $1-\delta_0$ on distribution $\sigma$,
\begin{eqnarray}\label{eqn:icost}
I(X, Y ; \Pi | Z, P) \geq \frac{Cm}{B^2},
\end{eqnarray}
where, with some abuse of notation, 
$\Pi$ is also used to denote the transcript of the corresponding randomized protocol,
and here the input $(X,Y)$ is drawn
from $\sigma$ conditioned on $(X,Y)$ being a NO instance. 
Here, $\Pi$ is randomized, and succeeds with probability at least $1-\delta_0$,
where the probability is over the joint space of the random coins of
$\Pi$ and the input distribution.

Our starting point for proving (\ref{eqn:icost}) is 
Jayram's lower bound for the conditional mutual information when the inputs
are drawn from a related distribution 
(reference [70] on p.182 of \cite{BarYossefThesis}),
but we require several non-trivial modifications to his argument 
in order to apply it to bound the conditional mutual information for our input distribution, 
which is $\sigma$ conditioned
on $(X,Y)$ being a NO instance. Essentially, we are able to show that the variation
distance between our distribution and his distribution is small, and use this to
bound the difference in the conditional mutual information between the two distributions.
The proof is rather technical, and we postpone it to Appendix \ref{app:dist}. 

%
We make a few simple refinements to (\ref{eqn:icost}).  Define the random variable
$W$ which is $1$ if $(X,Y)$ is a YES instance, and $0$ if $(X,Y)$ is a
NO instance. Then by definition of the mutual information, 
if $(X,Y)$ is drawn from
$\sigma$ without conditioning on $(X,Y)$ being a NO instance, then we have
\begin{eqnarray*}
I(X, Y ; \Pi | W, Z, P) & \geq & \frac{1}{2} I(X, Y ; \Pi | Z, P, W = 0)\\
& = & \Omega(m/B^2).
\end{eqnarray*}
Observe that
\begin{align}\label{eqn:mutInf}
I(X, Y ; \Pi | S, W, Z, P) &\geq I(X, Y ; \Pi | W, Z, P) - H(S) 
= \Omega(m/B^2),
\end{align}
where we assume that $\Omega(m/B^2) - \log m = \Omega(m/B^2)$. 
Define the constant $\delta_1 = \delta_0/4$. 
We now define a problem which involves solving $r$ copies of $\gaplinf^B$. 
%
\begin{definition}[$\multiLinf^{r,B}$ Problem]\label{def:multi} There are $r$ pairs of
  inputs $(x^1,y^1),(x^2,y^2),\ldots,(x^r,y^r)$ such that each pair
  $(x^i,y^i)$ is a legal instance of the $\gaplinf^B$ problem. Alice
  is given $x^1, \ldots, x^{r}$.  Bob is given $y^1, \ldots, y^r$.
  The goal is to output a vector $v \in \{NO,YES\}^r$, so that for at
  least a $1-\delta_1$ fraction of the entries $i$, 
  $v_i = \gaplinf^B(x^i, y^i)$.
\end{definition}
\begin{remark}
Notice that Definition \ref{def:multi} is defining a promise problem. We will 
study the distributional complexity of this problem under the distribution
$\sigma^r$, which is a product distribution on 
the $r$ instances $(x^1, y^1), (x^2, y^2), \ldots, (x^r, y^r)$. 
\end{remark}
\begin{theorem}\label{thm:theMain}
$D_{\sigma^r, \delta_1}(\multiLinf^{r,B}) = \Omega(rm/B^2).$
\end{theorem}
\begin{proof}
  Let $\Pi$ be any deterministic
  protocol for $\multiLinf^{r,B}$ which succeeds with probability at
  least $1-\delta_1$ in solving $\multiLinf^{r,B}$ when the inputs are
  drawn from $\sigma^r$, where the probability is taken over the input
  distribution. We show that $\Pi$ has communication cost
  $\Omega(rm/B^2)$.

  Let $X^1, Y^1, S^1, W^1, Z^1, P^1 \ldots, X^r, Y^r, S^r, W^r, Z^r,$ and $P^r$
  be the random variables associated with $\sigma^r$, i.e., 
  $X^j, Y^j, S^j, W^j, P^j$ and $Z^j$ correspond to the random variables
  $X, Y, S, W,Z,P$ associated with the $j$-th independent instance drawn
  according to $\sigma$, defined above. We let
  $X = (X^1, \ldots, X^r)$, $X^{< j} = (X^1, \ldots, X^{j-1})$, and
  $X^{-j}$ equal $X$ without $X^j$. Similarly we define these
  vectors for $Y, S, W, Z$ and $P$.
 
By the chain rule for mutual information, 
$I(X^1, \ldots, X^r, Y^1, \ldots, Y^r ; \Pi | S, W, Z, P)$
is equal to $\sum_{j=1}^r I(X^j, Y^j ; \Pi | X^{< j}, Y^{<j}, S, W, Z, P).$
Let $V$ be the output of $\Pi$, and $V_j$ be its $j$-th coordinate. 
For a value $j \in [r]$, we say that $j$ is {\it good} if
$\Pr_{X, Y}[V_j = \gaplinf^{B}(X^j, Y^j)] \geq 1-\frac{2\delta_0}{3}.$
Since $\Pi$ succeeds with probability at least $1-\delta_1 = 1-\delta_0/4$ in
outputting a vector with at least a $1-\delta_0/4$ fraction of
correct entries, the expected probability of success over a random $j \in [r]$
is at least $1-\delta_0/2$, and so by a
Markov argument, there are $\Omega(r)$ good indices $j$.

Fix a value of $j \in [r]$ that is good, and consider $I(X^j, Y^j ; \Pi
| X^{< j}, Y^{<j}, S, W, Z, P)$.  By expanding the conditioning, $I(X^j,
Y^j ; \Pi | X^{< j}, Y^{<j}, S, W, Z, P)$ is equal to
\begin{eqnarray}\label{eqn:condition}
{\bf E}_{x,y,s,w,z,p} [I(X^j , Y^j ; \Pi 
\mid (X^{< j}, Y^{<j}, S^{-j}, W^{-j}, Z^{-j}, P^{-j}) = 
(x, y, s, w, z, p), S^j, W^j, Z^j, P^j)].
\end{eqnarray}
For each $x,y,s,w,z,p$, define a randomized protocol
$\Pi_{x,y,s,w,z,p}$ for $\gaplinf^B$ under distribution $\sigma$.
Suppose that Alice is given $a$ and Bob is given $b$, where $(a,b) \sim
\sigma$. Alice sets $X^j = a$, while Bob sets $Y^j = b$. Alice and Bob
use $x,y,s,w,z$ and $p$ to set their remaining inputs as follows.
Alice sets $X^{<j} = x$ and Bob sets $Y^{<j} = y$.  Alice and
Bob can randomly set their remaining inputs without any communication, 
since for $j' > j$, conditioned on $S^{j'}, W^{j'},Z^{j'}$, and $P^{j'}$,
Alice and Bob's inputs are independent.  
Alice and Bob run $\Pi$ on
inputs $X, Y$, and define $\Pi_{x,y,s,w,z,p}(a,b) = V_j.$ 
We say a tuple $(x,y,s,w,z,p)$ is {\it good} if
\begin{align*}
\Pr_{X,Y}[V_j = \gaplinf^{B}(X^j,Y^j) \ \mid \ X^{<j} = x, Y^{<j} = y, 
S^{-j} = s, \ifconf \\ \fi W^{-j} = w, Z^{-j} = z, P^{-j} = p] \geq 1-\delta_0.
\end{align*}
By a Markov argument, and using that $j$ is good, we have 
$\Pr_{x,y,s,w,z,p}[(x,y,s,w,z,p) \textrm{ is good }] = \Omega(1).$
Plugging into (\ref{eqn:condition}), 
$I(X^j, Y^j ; \Pi | X^{< j}, Y^{<j}, S, W, Z, P)$ is at least a constant times
\begin{align*}
{\bf E}_{x,y,s,w,z,p} [I(X^j Y^j ; \Pi |
 &
(X^{< j},Y^{<j}, S^{-j},W^{-j},Z^{-j},P^{-j})  = (x, y, s, w, z, p),
 \\ &
S^j, W^j, Z^j, P^j, (x,y, s,w,z, p)
\textrm{ is good})].
\end{align*} For any $(x, y, s, w, z, p)$ that
is good, $\Pi_{x,y, s,w,z,p}(a,b) = V_j$ with
probability at least $1-\delta_0$, over the joint distribution of the
randomness of $\Pi_{x,y,s,w,z,p}$ and $(a,b) \sim \sigma$. 
By (\ref{eqn:mutInf}),
\begin{align*}
{\bf E}_{x,y,s,w,z,p} [I(X^j, Y^j ; \Pi | 
&
(X^{< j}, Y^{<j}, S^{-j}, W^{-j}, Z^{-j}, P^{-j}) = 
(x, y, s, w, z, p),
\\&
S^j, W^j, Z^j, P^j, (x,y,s,w,z,p) \textrm{ is good}]
= \Omega \left (\frac{m}{B^2} \right ).
\end{align*}
Since there are $\Omega(r)$ good indices $j$, we have 
$I(X^1, \ldots, X^r ; \Pi | S, W, Z, P) = \Omega(mr/B^2).$
Since the distributional complexity $D_{\sigma^r, \delta_1}(\multiLinf^{r,B})$ is at least the
minimum of $I(X^1, \ldots, X^r ; \Pi | S, W, Z, P)$ over 
deterministic protocols $\Pi$ which succeed with probability at least
$1-\delta_1$ on input distribution $\sigma^r$, it follows that
$D_{\sigma^r, \delta_1}(\multiLinf^{r,B}) = \Omega(mr/B^2)$.
\end{proof}

\subsection{The overall lower bound}\label{sec:cc2}
We use the theorem in the previous subsection with an extension of the method of section 6.3
of \cite{PW11}. 

  Let $X \subset \R^n$ be a distribution with $x_i \in \{-n^d, \dotsc,
  n^d\}$ for all $i \in [n]$ and $x \in X$.  Here $d = \Theta(1)$ is a parameter. 
Given an adaptive compressed 
  sensing scheme $\mathcal{A}$, we define a
  $(1+\eps)$-approximate $\ell_1/\ell_1$ sparse recovery
  {\it multiround bit scheme} on $X$ as follows. 

Let $A^i$ be the $i$-th (adaptively chosen) measurement matrix of the compressed sensing scheme. 
We may assume that the union of rows in matrices $A^1, \ldots, A^r$ generated by $\mathcal{A}$
is an orthonormal system, since the rows can be orthogonalized in a post-processing step. 
We can assume that $r \leq n$. 

  Choose a random $u \in \mathbb{R}^n$ 
from distribution $\mathcal{N}(0, \frac{1}{n^c} \cdot I_{n \times n})$,
  where $c = \Theta(1)$ is a parameter.  
  We require that the compressed sensing scheme outputs a valid result of 
  $(1+\eps)$-approximate recovery on $x+u$ with probability at least $1-\delta$, 
  over the choice of $u$ and its random coins. By Yao's minimax principle, we can fix the 
  randomness of the compressed sensing scheme and
  assume that the scheme is deterministic. 

  Let $B^{1}$ be the matrix
  $A^1$ with entries rounded to $t \log n$ bits for a parameter $t = \Theta(1)$. 
  We compute $B^1x$. Then, we compute 
  $B^1x + A^1u$. From this, we compute $A^2$, 
   using the algorithm specified by $\mathcal{A}$ as if $B^1x + A^1u$ were equal to 
   $A^1x'$ for some $x'$. For this, we use the following lemma, which is Lemma~5.1 of~\cite{DIPW}.
\begin{lemma}\label{lem:roundingFix}
Consider any $m \times n$ matrix $A$ with orthonormal rows. Let $B$ be the result
of rounding $A$ to $b$ bits per entry. Then for any $v \in \mathbb{R}^n$ there
exists an $s \in \mathbb{R}^n$ with $Bv = A(v-s)$ and $\|s\|_1 < n^2 2^{-b} \|v\|_1$.
\end{lemma}

  In general for $i \geq 2$, given $B^1x + A^1u, B^2x + A^2u, \ldots, B^{i-1}x + A^{i-1}u$ we 
  compute $A^i$, and round to $t\log n$
  bits per entry to get $B^i$.
  The output of the multiround bit scheme is the same as that of 
  the compressed sensing scheme. If the compressed sensing scheme uses $r$ rounds, then the multiround
  bit scheme uses $r$ rounds. Let $b$ denote the total number of bits in the concatenation of discrete
  vectors $B^1x, B^2x, \ldots, B^rx$. 

We give a generalization of Lemma 5.2 of \cite{PW11} which relates bit
schemes to sparse recovery schemes. Here we need to generalize the relation
from non-adaptive schemes to adaptive schemes, using Gaussian noise instead of uniform noise,
and arguing about multiple rounds of the algorithm. 

\begin{lemma}\label{lem:bits}
  For $t = O(1 + c +d)$, a lower bound of $\Omega(b)$ bits for a multiround bit scheme 
  with error probability at most $\delta+1/n$ implies a lower bound of
  $\Omega(b/((1+c+d)\log n))$ measurements for 
  $(1+\eps)$-approximate sparse recovery schemes with failure probability
  at most $\delta$.
\end{lemma}
\begin{proof}
  Let $\mathcal{A}$ be a $(1+\eps)$-approximate adaptive compressed
  sensing scheme with failure probability $\delta$. We will show that
  the associated multiround bit scheme has failure probability $\delta
  + 1/n$.
%

By Lemma \ref{lem:roundingFix}, for any vector $x \in \{-n^d, \ldots, n^d\}$ 
we have $B^1x = A^1(x+s)$ for a vector $s$ with
  $\norm{1}{s} \leq n^22^{-t\log n}\norm{1}{x}$, so $\norm{2}{s} \leq
  n^{2.5-t} \norm{2}{x} \leq n^{3.5+d-t}$. 
Notice that $u+s \sim \mathcal{N}(s, \frac{1}{n^c} \cdot I_{n \times n})$.
We use the following quick suboptimal upper bound on the 
statistical distance between two univariate normal
distributions, which suffices for our purposes. 
\begin{fact}(see section 3 of \cite{p05})\label{fact:gaussian}
The variation distance between $\mathcal{N}(\theta_1, 1)$ and $\mathcal{N}(\theta_2,1)$ is $\frac{4\tau}{\sqrt{2\pi}} + O(\tau^2),$
where $\tau = |\theta_1 - \theta_2|/2$. 
\end{fact}
It follows by Fact \ref{fact:gaussian} and independence across coordinates, that the variation
distance between $\mathcal{N}(0, \frac{1}{n^c} \cdot I_{n \times n})$ and 
$\mathcal{N}(s, \frac{1}{n^c} \cdot I_{n \times n})$ is the same as that between
$\mathcal{N}(0, I_{n \times n})$ and $\mathcal{N}(s \cdot n^{c/2}, I_{n \times n})$, which can be upper-bounded
as
\begin{eqnarray*}
\sum_{i=1}^n \cdot \frac{2n^{c/2} |s_i|}{\sqrt{2\pi}} + O(n^c s_i^2) & = & O(n^{c/2} \|s\|_1 + n^c \|s\|_2^2)\\ 
& = & O(n^{c/2} \cdot \sqrt{n} \|s\|_2 + n^c \|s\|_2^2)\\
 & = & O(n^{c/2 + 4+d-t} + n^{c + 7+2d-2t}).
\end{eqnarray*}
%
It follows that for $t = O(1 + c + d)$, the variation distance is at most $1/n^2$. 

  Therefore, if $\mathcal{T}^1$ is the algorithm which
  takes $A^1(x+u)$ and produces $A^2$, then 
  $\mathcal{T}^1(A^1(x+u)) = \mathcal{T}^1(B^1x + A^1u)$ with probability
  at least $1-1/n^2$. This follows since $B^1x+A^1u = A^1(x + u + s)$
  and $u+s$ and $u$ have variation distance at most $1/n^2$. 

In the second round, $B^2 x + A^2u$ is obtained,
  and importantly we have for the algorithm $\mathcal{T}^2$ in the second round,
  $\mathcal{T}^2(A^2(x+u)) = \mathcal{T}^2(B^2x + A^2u)$ with probability
  at least $1-1/n^2$. This follows
  since $A^2$ is a deterministic function of $A^1u$, and $A^1u$ and $A^2u$ are independent
  since $A^1$ and $A^2$ are orthonormal while $u$ is a vector of i.i.d. Gaussians (here we
  use the rotational invariance / symmetry of Gaussian space). 
  It follows by induction that with probability at least $1-r/n^2 \geq 1-1/n$, the output
  of the multiround bit scheme agrees with that of $\mathcal{A}$ on input $x + u$.

  Hence, if $m_i$ is the number of measurements in round $i$, and $m =
  \sum_{i=1}^r m_i$, then we have a multiround bit scheme using a
  total of $b = mt\log n = O(m(1 + c+d) \log n)$ bits and with failure
  probability $\delta + 1/n$.
%
\end{proof}
The rest of the proof is similar to the proof of the non-adaptive lower bound for $\ell_1/\ell_1$
sparse recovery given in \cite{PW11}. We sketch the proof, referring the reader to \cite{PW11}
for some of the details.
Fix parameters $B = \Theta(1/\eps^{1/2})$, $r = k$, $m =
1/\eps^{3/2}$, and $n = k/\eps^3$. Given an instance $(x^1, y^1),
\ldots, (x^r, y^r)$ of $\multiLinf^{r,B}$ we define the input signal
$z$ to a sparse recovery problem. We allocate a set $S^i$ of $m$
disjoint coordinates in a universe of size $n$ for each pair $(x^i,
y^i)$, and on these coordinates place the vector $y^i-x^i$.  The
locations turn out to be essential for the proof of Lemma \ref{lem:noise} below, 
and are placed uniformly at random among the $n$ total coordinates (subject to the
constraint that the $S^i$ are disjoint). 
Let $\rho$ be the induced distribution on $z$.

Fix a $(1+\eps)$-approximate $k$-sparse recovery multiround bit scheme $Alg$ that
uses $b$ bits and succeeds with
probability at least $1-\delta_1/2$ over $z \sim \rho$.  
Let $S$ be the set of top
$k$ coordinates in $z$. As shown in equation (14) of \cite{PW11},
$Alg$ has the guarantee that if $w = Alg(z)$, then
\begin{eqnarray}\label{eqn:main}
\|(w-z)_S\|_1 + \|(w-z)_{[n] \setminus S}\|_1 \leq (1+2\eps)\|z_{[n] \setminus S}\|_1.
\end{eqnarray}
(the $1+2\eps$ instead of the $1+\eps$ factor is to handle the rounding of entries of
the $A^i$ and the noise vector $u$). 
Next is our generalization of Lemma 6.8 of \cite{PW11}.

\begin{lemma}\label{lem:oneSide}
For $B = \Theta(1/\eps^{1/2})$ sufficiently large, suppose that 
$$\Pr_{z \sim \rho}[\|(w-z)_S\|_1 \leq 10\eps \cdot \|z_{[n] \setminus S}\|_1] \geq 1-\frac{\delta_1}{2}.$$ 
Then $Alg$ requires $b = \Omega(k/\eps^{1/2})$. 
\end{lemma}
\begin{proof}
We show how to use $Alg$ to solve instances of $\multiLinf^{r,B}$
  with probability at least $1-\delta_1$, where the probability is
  over input instances to $\multiLinf^{r,B}$ distributed according to
  $\sigma^r$, inducing the distribution $\rho$ on $z$. The lower bound will
  follow by Theorem \ref{thm:theMain}. Let $w$ be the output of 
  $Alg$. 

  Given $x^1, \ldots, x^r$, Alice places $-x^i$ on the appropriate
  coordinates in the set $S^i$ used in defining $z$, obtaining a
  vector $z_{Alice}$. 
  Given $y^1, \ldots, y^r$, Bob places the $y^i$ on the appropriate
  coordinates in $S^i$. He thus creates a vector $z_{Bob}$ for which
  $z_{Alice} + z_{Bob} = z$. In round $i$, Alice transmits $B^iz_{Alice}$
  to Bob, who computes $B^i(z_{Alice} + z_{Bob})$ and transmits it back 
  to Alice. Alice can then compute $B^i(z) + A^i(u)$ for a random
  $u \sim \mathcal{N}(0, \frac{1}{n^c} \cdot I_{n \times n})$. 
  We can assume all
  coordinates of the output vector $w$ are in the real interval $[0,B]$, since rounding the
  coordinates to this interval can only decrease the error.

  To continue the analysis, we use a proof technique of \cite{PW11} (see the proof
of Lemma 6.8 of \cite{PW11} for a comparison). 
For each $i$ we say that $S^i$ is {\it bad} if either
\begin{itemize}
\item there is no coordinate $j$ in $S^i$ for which $|w_j| \geq \frac{B}{2}$
yet $(x^i, y^i)$ is a YES instance of $\gaplinf^B$, or
\item there is a coordinate $j$ in $S^i$ for which $|w_j| \geq \frac{B}{2}$
yet either $(x^i, y^i)$ is a NO instance of $\gaplinf^B$ or $j$ is
not the unique $j^*$ for which $y_{j^*}^i - x_{j^*}^i = B$. 
\end{itemize}
  The proof of Lemma 6.8 of \cite{PW11} shows that the fraction $C > 0$ of bad $S^i$ can be made
  an arbitrarily small constant by appropriately choosing an appropriate $B = \Theta(1/\eps^{1/2})$.
  Here we choose $C = \delta_1$. We also condition on $\|u\|_2\leq n^{-c}$ for a sufficiently
  large constant $c > 0$, which occurs
  with probability at least $1-1/n$. Hence, with probability at least $1-\delta_1/2-1/n > 1-\delta_1$,
  we have a $1-\delta_1$ fraction of indices $i$ for which the following algorithm
  correctly outputs $\gaplinf(x^i, y^i)$: if there is a $j \in S^i$
  for which $|w_j| \geq B/2$, output YES, otherwise output NO. 
  It follows by Theorem \ref{thm:theMain} that $Alg$
  requires $ b = \Omega(k/\eps^{1/2})$, independent of the number
  of rounds.
\end{proof}

The next lemma is the same as Lemma 6.9 of \cite{PW11}, replacing
$\delta$ in the lemma statement there with the constant $\delta_1$ and
observing that the lemma holds for compressed sensing schemes with an
arbitrary number of rounds. 
\begin{lemma}\label{lem:noise}
  Suppose $\Pr_{z \sim \rho}[\|(w-z)_{[n] \setminus S}\|_1 \leq
  (1-8\eps) \cdot \|z_{[n] \setminus S}\|_1] \geq \delta_1.$  Then
  $Alg$ requires $b = \Omega(k \log(1/\eps) / \eps^{1/2})$.
\end{lemma}
\begin{proof}
  As argued in Lemma 6.9 of \cite{PW11}, we have $I(w ; z)
  = \Omega(\eps mr \log(n/(mr)))$, which implies that $b = \Omega(\eps
  mr \log(n/(mr)))$, independent of the number $r$ of rounds used by
  $Alg$, since the only information about the signal is in the concatenation
  of $B^1z, \ldots, B^rz$. 
\end{proof}
Finally, we combine our Lemma \ref{lem:oneSide} and Lemma
\ref{lem:noise} to prove the analogue of Theorem 6.10 of \cite{PW11},
which completes this section. 
\begin{theorem}\label{thm:conclusion}
  Any $(1+\eps)$-approximate $\ell_1/\ell_1$ recovery scheme with
  success probability at least $1-\delta_1/2-1/n$ must make
  $\Omega(k/(\eps^{1/2} \cdot \log (k/\eps)))$ measurements.
\end{theorem}
\begin{proof}
  We will lower bound the number of bits used by any $\ell_1/\ell_1$ multiround bit scheme
  $Alg$. If $Alg$ succeeds with probability at least $1-\delta_1/2$,
  then in order to satisfy (\ref{eqn:main}), we must either have
  $\|(w-z)_S\|_1 \leq 10\eps \cdot \|z_{[n] \setminus S}\|_1$ or
  $\|(w-z)_{[n] \setminus S}\|_1 \leq (1-8\eps)\|z_{[n] \setminus
    S}\|_1$. Since $Alg$ succeeds with probability at least
  $1-\delta_1/2$, it must either satisfy the hypothesis of Lemma
  \ref{lem:oneSide} or Lemma \ref{lem:noise}. But by these two lemmas,
  it follows that $b = \Omega(k/\eps^{1/2})$. Therefore by Lemma
  \ref{lem:bits}, any $(1+\eps)$-approximate $\ell_1/\ell_1$ sparse
  recovery algorithm succeeding with probability at least $1-\delta_1/2
  - 1/n$ requires $\Omega(k/(\eps^{1/2} \cdot \log(k/\eps)))$
  measurements.
\end{proof}

\section{Acknowledgements}

Some of this work was performed while E. Price was an intern at IBM
research, and the rest was performed while he was supported by an NSF
Graduate Research Fellowship.

\newpage
\bibliographystyle{alpha}
\bibliography{eps-sparse}

\appendix

\section{Relationship between Post-Measurement and Pre-Measurement
  noise}\label{app:noisekind}

In the setting of~\cite{ACD11}, the goal is to recover a $k$-sparse
$x$ from observations of the form $Ax+w$, where $A$ has unit norm rows
and $w$ is i.i.d. Gaussian with variance $\norm{2}{x}^2/\eps^2$.  By
ignoring the (irrelevant) component of $w$ orthogonal to $A$, this is
equivalent to recovering $x$ from observations of the form $A(x+w)$.
By contrast, our goal is to recover $x+w$ from observations of the
form $A(x+w)$, and for general $w$ rather than only for Gaussian $w$.

By arguments in~\cite{PW11,HIKP12}, for Gaussian $w$ the difference
between recovering $x$ and recovering $x+w$ is minor, so any lower
bound of $m$ in the~\cite{ACD11} setting implies a lower bound of
$\min(m, \eps n)$ in our setting.  The converse is only true for
proofs that use Gaussian $w$, but our proof fits this category.

\section{Information Chain Rule with Linear Observations}

\begin{lemma}\label{lemma:splitentropy}
  Suppose $a_i = b_i + w_i$ for $i \in [s]$ and the $w_i$ are
  independent of each other and the $b_i$.  Then
  \[
  I(a; b) \leq \sum_i I(a_i; b_i)
  \]
\end{lemma}
\begin{proof}
  Note that $h(a\mid b) = h(a-b\mid b) = h(w \mid b) = h(w)$.  Thus
  \begin{align*}
    I(a; b) &= h(a) - h(a \mid b) = h(a) - h(w)\\
    &\leq \sum_i h(a_i) - h(w_i) \\
    &= \sum_i h(a_i) - h(a_i \mid b_i) = \sum_i I(a_i; b_i)
  \end{align*}
\end{proof}

\section{Switching Distributions from Jayram's Distributional Bound}\label{app:dist}
We first sketch a proof of Jayram's lower bound on the distributional
complexity of $\gaplinf^B$ \cite{J02}, then change it to a different distribution that
we need for our sparse recovery lower bounds in Subsection \ref{sec:change}. Let 
$X, Y \in \{0, 1, \ldots, B\}^m$. Define distribution
$\mu^{m,B}$ as follows: for each $j \in
[m]$, choose a random pair $(Z_j, P_j) \in \{0, 1, 2, \ldots, B\} \times \{0,1\}
\setminus \{(0,1), (B, 0)\}$. If $(Z_j, P_j) = (z, 0)$, then $X_j = z$ and 
$Y_j$ is uniformly distributed in $\{z, z+1\}$; if $(Z_j, P_j) = (z, 1)$, then $Y_j = z$
and $X_j$ is uniformly distributed on $\{z-1, z\}$. 
Let $X = (X_1, \ldots, X_m)$, $Y = (Y_1, \ldots, Y_m)$, $Z = (Z_1, \ldots, Z_m)$ and $P = (P_1, \ldots, P_m)$. 

The other distribution we define is $\sigma^{m,B}$, 
which is the same as distribution $\sigma$ 
in Section \ref{sec:eps} (we include $m$ and $B$ in the
notation here for clarity). This is defined by
first drawing $X$ and $Y$ according to distribution $\mu^{m,B}$. Then, 
we pick a random coordinate $S \in [m]$ and replace
$(X_S, Y_S)$ with a uniformly random element in the set $\{(0,0), (0,B)\}$. 

Let $\Pi$ be a deterministic protocol that errs with probability at most $\delta$
on input distribution $\sigma^{m,B}$. 

By the chain rule for mutual information, when $X$ and $Y$ are
distributed according to $\mu^{m,B}$, 
\begin{eqnarray*}
I(X, Y ; \Pi | Z, P)
& = & \sum_{j=1}^m I(X_j, Y_j ; \Pi | X^{<j}, Y^{<j}, Z, P),
\end{eqnarray*}
which is equal to 
\begin{eqnarray*}
\sum_{j=1}^m {\bf E}_{x, y, z,p} [I(X_j, Y_j ; \Pi \ |
Z_j, P_j, X^{<j} = x, Y^{<j} = y, Z^{-j} = z, P^{-j} = p)].
\end{eqnarray*}
Say that an index $j \in [m]$ is {\it good} if conditioned on $S=j$,
$\Pi$ succeeds on $\sigma^{m,B}$ with probability at least $1-2\delta$.
By a Markov argument, at least $m/2$ of the indices $j$ are good. 
Fix a good index $j$. 

We say that the tuple $(x, y, z, p)$ is {\it good} if conditioned on $S = j$,
$X^{<j} = x$, $Y^{<j} = y$, $Z^{-j} = z$, and $P^{-j} = p$, $\Pi$ succeeds on $\sigma^{m,B}$
with probability at least $1-4\delta$. By a Markov bound, with probability
at least $1/2$, $(x, y, z, p)$ is good. Fix a good $(x,y,z,p)$.

We can define a single-coordinate protocol $\Pi_{x,y,z,p,j}$ as follows. 
The parties use $x$ and $y$ to fill in their input
vectors $X$ and $Y$ for coordinates $j' < j$. They also use 
$Z^{-j} = z$, $P^{-j} = p$, and private randomness to fill in their inputs without any
communication on the remaining coordinates $j' > j$. They place their
single-coordinate input $(U, V)$ on their $j$-th coordinate. The parties
then output whatever $\Pi$ outputs. 

It follows that $\Pi_{x,y,z,p,j}$ is a 
single-coordinate protocol $\Pi'$ which distinguishes
$(0,0)$ from $(0,B)$ under the uniform distribution 
with probability at least $1-4\delta$. 
For the single-coordinate problem, we need to bound
$I(X_j, Y_j ; \Pi' | Z_j, P_j)$ when $(X_j, Y_j)$ is uniformly random from 
the set $\{(Z_j, Z_j), (Z_j, Z_j+1)\}$ if $P_j = 0$, and $(X_j, Y_j)$ is uniformly
random from the set $\{(Z_j, Z_j), (Z_j-1, Z_j)\}$ if $P_j = 1$. 
By the same argument as in the proof of Lemma 8.2 of 
\cite{BJKS04}, if $\Pi'_{u,v}$ denotes the distribution on transcripts
induced by inputs $u$ and $v$ and private coins, then we have
\begin{eqnarray}\label{eqn:final}
I(X_j, Y_j ; \Pi' | Z_j, P_j)
\geq \Omega(1/B^2) \cdot (h^2(\Pi'_{0,0}, \Pi'_{0, B}) + h^2(\Pi'_{B,0},\Pi'_{B, B})),
\end{eqnarray}
where 
$$h(\alpha, \beta) = \sqrt{\frac{1}{2}\sum_{\omega \in \Omega} (\sqrt{\alpha(\omega)} - \sqrt{\beta(\omega)})^2}$$
is the Hellinger distance between distributions $\alpha$ and $\beta$ on support $\Omega$.
For any two distributions $\alpha$ and $\beta$, if we define
$$D_{TV}(\alpha, \beta) = \frac{1}{2} \sum_{\omega \in \Omega} |\alpha(\omega)-\beta(\omega)|$$
to be the variation distance between them, then $\sqrt{2} \cdot h(\alpha, \beta) \geq
D_{TV}(\alpha, \beta)$ (see Proposition 2.38 of \cite{BarYossefThesis}).
 
Finally, since $\Pi'$ 
succeeds with probability at least $1-4\delta$ on the uniform distribution
on input pair in $\{(0,0), (0,B)\}$, we have
$$\sqrt{2} \cdot h(\Pi'_{0,0},\Pi'_{0, B}) \geq D_{TV}(\Pi'_{0,0},\Pi'_{0, B}) = \Omega(1).$$
Hence,
\begin{eqnarray*}
I(X_j, Y_j ; \Pi | Z_j, P_j, X^{<j} = x, Y^{<j} = y, Z^{-j} = z, P^{-j} = p)\\
 = \Omega(1/B^2)
\end{eqnarray*}
for each of the $\Omega(m)$ good $j$.  Thus $I(X, Y ; \Pi | Z, P) =
\Omega(m/B^2)$ when inputs $X$ and $Y$ are distributed according to
$\mu^{m,B}$, and $\Pi$ succeeds with probability at least $1-\delta$
on $X$ and $Y$ distributed according to $\sigma^{m,B}$.

\subsection{Changing the distribution}\label{sec:change} Consider the
distribution $$\zeta^{m,B} = (\sigma^{m,B} \mid (X_S, Y_S) =
(0,0)).$$ We show $I(X, Y ; \Pi | Z) = \Omega(m/B^2)$ when $X$ and $Y$
are distributed according to $\zeta^{m,B}$ rather than according to
$\mu^{m,B}$.

For $X$ and $Y$ distributed according to $\zeta^{m,B}$, by the chain rule we again have
that $I(X, Y ; \Pi | Z, P)$ is equal to 
\begin{eqnarray*}
\sum_{j=1}^m {\bf E}_{x, y, z, p} [I(X_j, Y_j ; \Pi |
Z_j, P_j, X^{<j} = x, Y^{<j} = y, Z^{-j} = z, P^{-j} = p)].
\end{eqnarray*}
Again, say that an index $j \in [m]$ is good if conditioned on $S=j$,
$\Pi$ succeeds on $\sigma^{m,B}$ with probability at least $1-2\delta$.
By a Markov argument, at least $m/2$ of the indices $j$ are good. 
Fix a good index $j$. 

Again, we say that the tuple $(x, y, z, p)$ is good if conditioned on $S = j$,
$X^{<j} = x$, $Y^{<j} = y$, $Z^{-j} = z$ and $P^{-j} = p$, $\Pi$ succeeds on $\sigma^{m,B}$
with probability at least $1-4\delta$. By a Markov bound, with probability
at least $1/2$, $(x, y, z, p)$ is good. Fix a good $(x,y,z,p)$.

As before, we can define a single-coordinate protocol $\Pi_{x,y,z,p,j}$. 
The parties use $x$ and $y$ to fill in their input
vectors $X$ and $Y$ for coordinates $j' < j$. They can also use 
$Z^{-j} = z$, $P^{-j} =p$, and private randomness to fill in their inputs without any
communication on the remaining coordinates $j' > j$. They place their
single-coordinate input $(U, V)$, uniformly drawn from $\{(0,0), (0,B)\}$,
on their $j$-th coordinate. The parties output whatever $\Pi$ outputs. Let $\Pi'$
denote $\Pi_{x,y,z,p,j}$ for notational convenience.

The first issue is that unlike before $\Pi'$ is not guaranteed to have success probability
at least $1-4\delta$ since $\Pi$ is not being run on input distribution $\sigma^{m,B}$
in this reduction. The second issue is in bounding $I(X_j, Y_j ; \Pi' | Z_j,P_j)$ since 
$(X_j, Y_j)$ is now drawn from the marginal distribution of $\zeta^{m,B}$ on coordinate $j$. 

Notice that $S \neq j$ with probability $1-1/m$, which we condition on. This immediately resolves
the second issue since now the marginal distribution on $(X_j, Y_j)$ is the same under $\zeta^{m,B}$
as it was under $\sigma^{m,B}$; namely
it is the following distribution: $(X_j, Y_j)$ is uniformly random from 
the set $\{(Z_j, Z_j), (Z_j, Z_j+1)\}$ if $P_j = 0$, and $(X_j, Y_j)$ is uniformly
random from the set $\{(Z_j, Z_j), (Z_j-1, Z_j)\}$ if $P_j = 1$. 

We now address the first issue.  After conditioning on $S \neq j$, we
have that $(X^{-j}, Y^{-j})$ is drawn from $\zeta^{m-1,B}$.  If
instead $(X^{-j}, Y^{-j})$ were drawn from $\mu^{m-1,B}$, then after
placing $(U, V)$ the input to $\Pi$ would be drawn from $\sigma^{m,B}$
conditioned on a good tuple.  Hence in that case, $\Pi'$ would succeed
with probability $1-4\delta$.  Thus for our actual distribution on
$(X^{-j}, Y^{-j})$, after conditioning on $S \neq j$, the success
probability of $\Pi'$ is at least
$$1-4\delta - D_{TV}(\mu^{m-1,B}, \zeta^{m-1,B}).$$

Let $C^{\mu, m-1, B}$ be the random variable which counts the number of 
coordinates $i$ for which $(X_i, Y_i) = (0,0)$ when $X$ and $Y$ are 
drawn from $\mu^{m-1,B}$. Let $C^{\zeta, m-1, B}$ be a random
variable which counts the number of coordinates $i$ for which
$(X_i, Y_i) = (0,0)$ when $X$ and $Y$ are drawn from $\zeta^{m-1, B}$. 
Observe that $(X_i, Y_i) = (0,0)$ in $\mu$ only if $P_i = 0$ and $Z_i = 0$, which happens
with probability $1/(2B)$. Hence,
$C^{\mu, m-1, B}$ is distributed as Binomial$(m-1, 1/(2B))$, while
$C^{\zeta, m-1, B}$ is distributed as Binomial$(m-2, 1/(2B)) + 1$. We use
$\mu'$ to denote the distribution of $C^{\mu, m-1, B}$ and $\zeta'$ to denote
the distribution of $C^{\zeta, m-1, B}$. Also, let $\iota$ denote the 
Binomial$(m-2, 1/(2B))$ distribution. 
Conditioned on $C^{\mu, m-1, B} = C^{\zeta, m-1, B}$, we have that $\mu^{m-1,B}$ and $\zeta^{m-1,B}$ are
equal as distributions, and so 
$$D_{TV}(\mu^{m-1,B}, \zeta^{m-1,B}) \leq D_{TV}(\mu', \zeta').$$
We use the following fact:
\begin{fact}\label{fact:vd} (see, e.g., Fact 2.4 of \cite{gmrz11}). 
Any binomial distribution $X$ with variance equal to $\sigma^2$ 
satisfies $D_{TV}(X, X+1) \leq 2/\sigma$. 
\end{fact}
By definition, $$\mu' = (1-1/(2B)) \cdot \iota + 1/(2B) \cdot \zeta'.$$
Since the variance of the Binomial$(m-2, 1/(2B))$ distribution is 
$$(m-2)/(2B) \cdot (1-1/(2B)) = m/(2B) (1-o(1)),$$ applying
Fact \ref{fact:vd} we have
\begin{eqnarray*}
D_{TV}(\mu', \zeta')
& = & D_{TV}((1-1/(2B)) \cdot \iota + (1/(2B)) \cdot \zeta', \zeta')\\
& = & \frac{1}{2} \cdot \|(1-1/(2B)) \cdot \iota + (1/(2B)) \cdot \zeta' - \zeta'\|_1\\
& = & (1-1/(2B)) \cdot D_{TV}(\iota, \zeta')\\
& \leq & \frac{2\sqrt{2B}}{\sqrt{m}} \cdot (1+o(1))\\
& = & O \left (\sqrt{\frac{B}{m}} \right ). 
\end{eqnarray*}
It follows that the success probability of $\Pi'$ is at least 
$$1-4\delta-O \left (\sqrt{\frac{B}{m}} \right ) \geq 1-5\delta.$$
Let $E$ be an indicator random variable for the event that $S \neq j$. Then
$H(E) = O((\log m) / m)$. Hence,
\begin{eqnarray*}
I(X_j, Y_j ; \Pi' | Z_j, P_j) & \geq & I(X_j, Y_j ; \Pi' | Z_j, P_j, E) - O((\log m) / m)\\
& \geq & (1-1/m) \cdot I(X_j, Y_j ; \Pi' | Z_j, P_j, S \neq j)
- O((\log m) / m)\\
& = & \Omega(1/B^2), 
\end{eqnarray*}
where we assume that $\Omega(1/B^2) - O((\log m) / m) = \Omega(1/B^2).$

Hence, $I(X, Y ; \Pi | Z, P) = \Omega(m/B^2)$ when inputs
$X$ and $Y$ are distributed according to $\zeta^{m,B}$, and $\Pi$ succeeds
with probability at least $1-\delta$ on $X$ and $Y$ distributed
according to $\sigma^{m,B}$. 


\end{document}